\documentclass[11pt, letterpaper]{article}
\usepackage{CJKutf8}

\usepackage[whole]{bxcjkjatype}

\usepackage{algorithm}
\usepackage[noend]{algorithmic}

\usepackage{cite} 
\usepackage[pdftex]{graphicx}
\usepackage{threeparttable}
\usepackage[margin=1in]{geometry}
\usepackage{amsmath,amssymb,amsfonts,setspace} 
\usepackage{amsthm} 
\usepackage{latexsym}
\usepackage{comment}
\usepackage{bm}
\usepackage{enumerate}
\usepackage{bookmark} 
\usepackage{type1cm} 
\usepackage{color} 
\usepackage{subcaption}

\usepackage{latexsym}
\usepackage{amsmath}
\usepackage{amsthm}
\usepackage{comment}
\usepackage{color}
\usepackage{graphicx}

\newcommand{\C}{}

\theoremstyle{plain}
\newtheorem{theorem}{Theorem}
\newtheorem{lemma}{Lemma}
\newtheorem{corollary}{Corollary} 
\newtheorem{definition}{Definition}
\newtheorem{proposition}{Proposition}

\title{Independent Set Reconfiguration Under Bounded-Hop Token Jumping
\author{ Hiroki Hatano
\and Naoki Kitamura
\and Taisuke Izumi
\and Takehiro Ito 
\and Toshimitsu Masuzawa
}
} 
\newcommand{\Isreconf}{\textsf{ISReconf}}
\newcommand{\ShortestIsreconf}{\textsf{Shortest-ISReconf}}
\newcommand{\ShortestunderkIsreconf}{\textsf{Shortest-ISReconf} under the $k$-jump rule}

\newcommand{\transition}[1]{\stackrel{#1}{\rightleftharpoons}}
\newcommand{\neighbor}[1]{\stackrel{#1}{\leftrightarrow}}
\newcommand{\dist}{\mathsf{dist}}

\newcommand{\symdiff}{\triangle}

\newcounter{cntLemmaNumber}

\newcounter{cntTheoremNumber}
\newenvironment{rtheorem}[1]{%
\setcounter{cntTheoremNumber}{\thetheorem}
\setcounterref{theorem}{#1}
\addtocounter{theorem}{-1}
\begin{theorem}
}{%
\end{theorem}
\setcounter{theorem}{\thecntTheoremNumber}
}
\newcounter{cntPropositionNumber}

\begin{document}

\date{}

\maketitle

\begin{abstract}
The independent set reconfiguration problem (\Isreconf) is the problem of determining, for two given independent sets of a graph, whether one can be transformed into the other by repeatedly applying a prescribed reconfiguration rule. 
There are two well-studied reconfiguration rules, called the Token Sliding (TS) rule and the Token Jumping (TJ) rule, and it is known that the complexity status of \Isreconf\ differs between the TS and TJ rules for some graph classes.
In this paper, we analyze how changes in reconfiguration rules affect the computational complexity of \Isreconf.
To this end, we generalize the TS and TJ rules to a unified reconfiguration rule, called the $k$-Jump rule, which removes one vertex from a current independent set and adds a vertex within distance $k$ from the removed vertex to obtain another independent set having the same cardinality. 
We give the following three results: 
First, we show that the reconfigurability of any \Isreconf\ instance does not change for all $k \ge 3$. 
Second, we present a polynomial-time algorithm to solve \Isreconf\ under the $2$-Jump rule for split graphs. 
Third, we consider the shortest variant of \Isreconf, which determines whether there is a transformation of at most $\ell$ steps, for a given integer $\ell \ge 0$. 
We prove that this shortest variant under the $k$-Jump rule is NP-complete for chordal graphs of diameter at most $2k + 1$, for any $k \geq 3$. 
\end{abstract}

\section{Introduction}
    Combinatorial reconfiguration~\cite{ItoDHPSUU11,Nishimura18,Heuvel13} has received much attention in the field of discrete algorithms and the computational complexity theory. 
    A typical \emph{reconfiguration problem} requires us to determine whether there is a step-by-step transformation between two given feasible solutions of a combinatorial (search) problem such that all intermediate solutions are also feasible and each step respects a prescribed reconfiguration rule. 
    This type of reconfiguration problems have been studied actively for several well-known feasible solutions on graphs, such as independent sets, cliques, vertex covers, colorings, matchings,~etc. (See surveys~\cite{Nishimura18,Heuvel13}.) 

    While reconfiguration problems have been considered for a wide range of feasible solutions, there are no clear rules to define reconfiguration rules; 
the smallest change to a current solution is often adopted as the reconfiguration rule unless there is a motivation from the application side. 
    To the best of our knowledge, even the most well-known reconfiguration rules, called the Token Jumping and Token Sliding rules, have no clear justification for their plausibility, which evokes the interest to more divergent reconfiguration rules beyond them.
    In this paper, we study and analyze how changes in reconfiguration rules affect the computational complexity of reconfiguration problems.

\subsection{Reconfiguration Rules and Related Known Results}

In this paper, we consider the reconfiguration problem for independent sets of a graph~\cite{KaminskiMM12,BousquetMNS22}, which is one of the most well-studied reconfiguration problems.
A subset of vertices of a graph $G = (V, E)$ is called an \emph{independent set} if no two vertices in the subset are adjacent in $G$.
In the context of reconfiguration problems,
an independent set is often interpreted to the placement of a set of \emph{tokens}, i.e., we regard an independent set 
$I \subseteq V$ as the locations of $|I|$ tokens in the graph. 
Then, one reconfiguration step of a independent set corresponds to the movement of a single token from some vertex to another vertex (on which no token is placed), such that the token locations after the movement also forms an independent set. Notice that the size of the independent set before and after the token movement remains unchanged. 
Reconfiguration rules define the allowed movements of tokens, and there are two well-used rules, called \emph{Token Sliding} and \emph{Token Jumping}~\cite{KaminskiMM12}:
\begin{itemize}
\item Token Sliding (TS) rule: a token is allowed to move only to a vertex adjacent to the current vertex; and
\item Token Jumping (TJ) rule: a token is allowed to move to an arbitrary vertex (on which no token is placed) in the graph. 
\end{itemize}

The independent set reconfiguration problem (\Isreconf) is now the problem to determine whether a given initial independent set $I_s$ of a graph $G$ can be transformed into a given target independent set $I_t$ of $G$ (having the same size as $I_s$) by moving tokens one by one under the prescribed reconfiguration rule while preserving the independence of the token placements during the transformation.  
The optimization problem, the shortest independent set reconfiguration problem (\ShortestIsreconf), of the above decision problem is also derived naturally: 
Given independent sets $I_s$ and $I_t$ ($|I_s|=|I_t|$) of a graph $G$, find the smallest number of token movements required to transform $I_s$ into $I_t$ under the prescribed reconfiguration rule, if exists.

Under both the TS and TJ rules, \Isreconf\ is known to be PSPACE-complete even for planar graphs of maximum degree three and bounded bandwidth~\cite{Zanden15}.
Therefore, algorithmic developments have been obtained for several restricted graph classes. (See the survey~\cite{BousquetMNS22} about \Isreconf.)
In particular, some known results show interesting contrasts of the complexity status between the TS and TJ rules, as follows:
For split graphs, \Isreconf\ is PSPACE-complete under the TS rule~\cite{BelmonteKLMOS21}, while it is solvable in polynomial time under the TJ rule~\cite{KaminskiMM12}.\footnote{Kami\'{n}ski et al.~\cite{KaminskiMM12} indeed gave a polynomial-time algorithm to solve \ShortestIsreconf\ under the TJ rule for even-hole-free graphs, which form a super graph class of split graphs.}
For bipartite graphs, \Isreconf\ is PSPACE-complete under the TS rule~\cite{LokshtanovM19}, while it is NP-complete under the TJ rule~\cite{LokshtanovM19}. 
The latter contrast on bipartite graphs implies that there is a yes-instance on bipartite graphs such that even a shortest transformation requires a super-polynomial number of steps under the TS rule, with the assumption of $\textup{NP} \neq \textup{PSPACE}$; on the other hand, any shortest transformation for bipartite graphs needs only a polynomial number of steps under the TJ rules. 

\subsection{Unified Reconfiguration Rule and Our Contributions}
The main purpose of our paper is to analyze how changes in reconfiguration rules affect the computational complexity of reconfiguration problems.
We regard the difference between the TS and TJ rules as ``movable distance'' of each token in one step: the TS rule allows a token to move to a vertex of distance one, while
the TJ rule allows a token to move to a vertex of distance at most $D(G)$, where $D(G)$ is the diameter of $G$.
From this perspective, we generalize the TS and TJ rules to a unified reconfiguration rule, called the \emph{$k$-Jump rule}, which allows a token to move to a vertex within distance $k$ from the current vertex, for an integer $k$, $1 \le k \le D(G)$.
Then, the TS rule is the $1$-Jump rule, and the TJ rule is the $D(G)$-Jump rule for a connected graph $G$. 

This generalization yields a natural question of what complexity landscape lies in the case of $1 < k < D(G)$, particularly for the graph classes exhibiting different computational complexity
between the TJ and TS rules. 
In this paper, we present three results addressing this question, which provides precise and interesting contrasts to the complexity status of (\textsf{Shortest}-)\Isreconf. 
Throughout this paper, let $G=(V, E)$ be an input graph, $I_s \subseteq V$ be an initial independent set, and $I_t \subseteq V$ be a target independent set.
We denote by a triple $(G, I_s, I_t)$ an instance of \Isreconf\ under the $k$-Jump rule.
We say that $(G, I_s, I_t)$ is \emph{reconfigurable}, if $I_s$ can be transformed into $I_t$ under the $k$-Jump rule. 

The first result shows that the reconfigurability of an instance $(G, I_s,I_t)$ does not change for any $k \ge 3$. 
Note that the following theorem holds for any connected graph $G$.
\begin{theorem}
\label{thm:SimulateTJ}
Let $G$ be a connected graph, and $k \geq 3$ be an arbitrary integer.
An instance $(G, I_s, I_t)$ is reconfigurable under the $k$-Jump rule if and only if $(G, I_s,I_t)$ is reconfigurable under the $D(G)$-Jump rule. 
\end{theorem}

    \begin{figure}[t]
        \centering
	\includegraphics[width=0.45\linewidth]{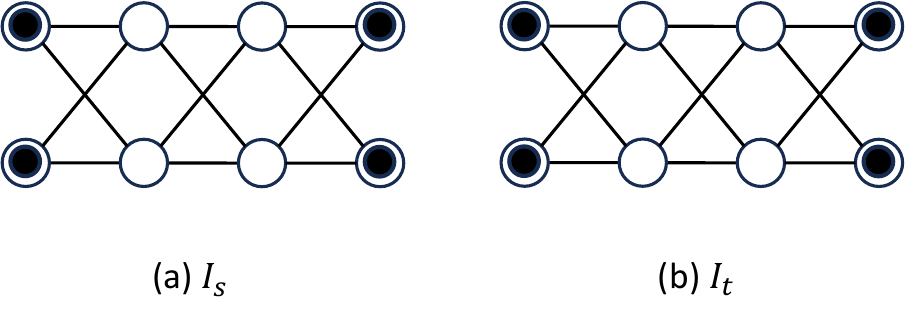}\\
	\caption{
		No-instance for split graphs under the $2$-Jump rule. Note that this is a yes-instance under the $k$-Jump rule, $k \ge 3$.
	}
	\label{fig:split}
    \end{figure}
    
While Theorem~\ref{thm:SimulateTJ} shows that the reconfigurability of an instance does not change for all $k \ge 3$, it may differ between $k \le 2$ and $k \ge 3$. 
For example, see \figurename~\ref{fig:split}, where $G$ is a split graph. 
For split graphs, any instance with $|I_s| = |I_t|$ is reconfigurable under the $k$-Jump rule, $k\ge 3$~\cite{KaminskiMM12}.
On the other hand, as we have seen in the example in \figurename~\ref{fig:split}, there exist instances for split graphs which are not reconfigurable under the $2$-Jump rule. 
Nonetheless, we give the following theorem, as our second result.
\begin{theorem}
\label{theorem2}
There exists a polynomial-time algorithm for \Isreconf\ under the $2$-Jump rule for split graphs.
\end{theorem}
Recall that \Isreconf\ under the $1$-Jump (i.e., TS) rule is PSPACE-complete for split graphs~\cite{BelmonteKLMOS21}. 
Thus, the complexity status of \Isreconf\ differs between $k=1$ and $k=2$.

Theorem~\ref{thm:SimulateTJ} says that the complexity status of \Isreconf\ is identical for all $k \ge 3$. 
Our third result shows that this does not hold for the optimization variant, \ShortestIsreconf. 
We note that \ShortestIsreconf\ under the $D(G)$-Jump rule is solvable in polynomial time for even-hole-free graphs~\cite{KaminskiMM12}, which include chordal graphs. 
\begin{theorem}
\label{theorem3}
Let $k \geq 3$ be any integer. 
Then, there exists a graph class $\mathcal{G}_k$ such that $\mathcal{G}_k$ is a subclass of chordal graphs of diameter at most $2k + 1$ and \ShortestIsreconf\ under the $k$-Jump rule is NP-complete for $\mathcal{G}_k$.
\end{theorem}

Tables~\ref{table:ISReconf} and~\ref{table:Shortest-ISReconf} summarize our results and the comparison with known results. 
Recall that split graphs form a subclass of chordal graphs and chordal graphs form a subclass of even-hole-free graphs.
We here give a remarks to Table~\ref{table:Shortest-ISReconf}. 
The complexity of {\ShortestIsreconf} under the $k$-Jump rule ($k \geq 3$) for split graphs is trivially obtained from the complexity of the problem under the $D(G)$-Jump rule for split graphs, because the diameter of any split graph is at most three.
\begin{table}[tbp]
    \centering
    \caption{Computational complexity of $\Isreconf$ for connected graphs $G$, where the term ``PSPACE-c'' is the abbreviation of PSPACE-complete.}
    \label{table:ISReconf}
    \medskip
    \begin{tabular}{c||c|c|c}
        \hline
        Reconfiguration Rules & Even-hole-free & Chordal & Split\\
        \hline
        TS ($1$-Jump) & PSPACE-c~\cite{BelmonteKLMOS21} & PSPACE-c~\cite{BelmonteKLMOS21} & PSPACE-c~\cite{BelmonteKLMOS21}\\
        $2$-Jump & open & open & P (Thm.~\ref{theorem2}) \\
        $k$-Jump ($k \geq 3$) & P (Thm.~\ref{thm:SimulateTJ}) & P (Thm.~\ref{thm:SimulateTJ}) & P (Thm.~\ref{thm:SimulateTJ}) \\
        TJ ($D(G)$-Jump) & P~\cite{KaminskiMM12} & P~\cite{KaminskiMM12} & P~\cite{KaminskiMM12} \\
        \hline
    \end{tabular}
\end{table}

\begin{table}[tbp]
    \centering
    \caption{Computational complexity of $\ShortestIsreconf$ for connected graphs $G$, where the term ``NP-c'' means NP-complete.}
    \label{table:Shortest-ISReconf}
    \medskip
    \begin{tabular}{c||c|c|c}
        \hline
        Reconfiguration Rules & Even-hole-free & Chordal & Split\\
        \hline
        TS (1-Jump) & PSPACE-c~\cite{BelmonteKLMOS21} & PSPACE-c~\cite{BelmonteKLMOS21} & PSPACE-c~\cite{BelmonteKLMOS21}\\
        \begin{tabular}{c}
        $k$-Jump\\
        ($3\leq k \leq \frac{D(G)-1}{2} $)
        \end{tabular}& NP-c (Thm.~\ref{theorem3}) & NP-c (Thm.~\ref{theorem3})& P (trivial) \\
        TJ ($D(G)$-Jump) & P~\cite{KaminskiMM12} & P~\cite{KaminskiMM12} & P~\cite{KaminskiMM12} \\
        \hline
    \end{tabular}
\end{table}

\subsection{Roadmap}
In Section~\ref{sec:preliminaries}, we introduce fundamental definitions, terminologies, and notations.
Sections~\ref{sec:Simulate}, \ref{sec:2-ISReconf}, and \ref{sec:NP-completeness} respectively provide 
the results of Theorems~\ref{thm:SimulateTJ}, \ref{theorem2}, and \ref{theorem3}. Finally, we conclude
the paper in Section~\ref{sec:conclusion} with several open problems related to our new rule.

\section{Preliminaries}
\label{sec:preliminaries}
In this section, we explain the notation used in this paper.
For sets $X$ and $Y$, their symmetric difference is defined as $X \symdiff Y = (X \cup Y) \setminus (X \cap Y)$. 

We consider only undirected graphs that are simple and connected\footnote{We assume graphs are connected for simplicity although the proposed algorithm works without the assumption.}.
For a graph $G$, we sometimes denote by $V(G)$ and $E(G)$ the vertex set and the edge set of $G$, respectively. 
The set of the vertices adjacent to $v$ in $G$ is denoted by $N_G(v)$, that is,  $N_G(v)=\{w\in V(G) \mid \{v,w\} \in E(G)\}$.
Let $\dist_G(u,v)$ denote the distance  in $G$ between vertices $u, v \in V(G)$.
A subset $S\subseteq V(G)$ is called an \emph{independent set of $G$} if no two vertices in $S$ are adjacent in $G$.
Let $\mathcal{I}(G)$ be the set of all independent sets of $G$.

Let $k$ be a positive integer. 
For two independent sets $I, I' \in \mathcal{I}(G)$, we write $I \neighbor{k} I'$ if and only if $|I \setminus I'| = |I' \setminus I| = 1$ and $\dist_G(u, v) \leq k$ for $u \in I \setminus I'$, $v \in I' \setminus I$. 
Let $\transition{k}$ be the transitive closure of $\neighbor{k}$.
Note that, from the definitions, $\neighbor{k}$ and $\transition{k}$ satisfy the symmetry.
For two independent sets $I, J \in \mathcal{I}(G)$, a sequence $\langle I_0,I_1, \ldots, I_{\ell} \rangle$ of independent sets of $G$ is called a \emph{reconfiguration sequence} between $I$ and $J$ (under the $k$-Jump rule) if $I_0 = I$, $I_\ell = J$, and $I_i \neighbor{k} I_{i+1}$ for all $i$, $0 \leq i < \ell$.
The \emph{length} of a reconfiguration sequence $\langle I_0,I_1, \ldots, I_{\ell} \rangle$ is defined to be $\ell$. 
We now define the independent set reconfiguration problem under the $k$-Jump rule ({$k$-\Isreconf}) and its optimization variant, as follows.
\begin{definition}
Given a graph $G$ and independent sets $I_s, I_t \in \mathcal{I}(G)$, the problem {$k$-\Isreconf} asks to determine whether $I_s \transition{k} I_t$ or not.
\end{definition}

\begin{definition}
Given a graph $G$, an integer $\ell'$ and independent sets $I_s, I_t \in \mathcal{I}(G)$, the problem {\ShortestunderkIsreconf} asks to determine whether $I_s \transition{k} I_t$ and the minimum length of any reconfiguration sequence between $I_s$ and $I_t$ is less than or equal to $\ell'$ or not.
\end{definition}

In the following, we use \emph{tokens} to represent the vertices in an independent set $I \in \mathcal{I}(G)$; 
imagine that a token is placed on each vertex in $I$.  
If an independent set $I'$ of $G$ can be obtained from another independent set $I$ such that $I' =(I \setminus \{u\}) \cup \{v\}$ ($u \in I\setminus I', v \in I'\setminus I$), we say that the token on $u$ \emph{moves} to $v$.
To specify the token on a vertex $u \in I$, we often say the token $u$.
For an independent set $I \in \mathcal{I}(G)$, we say a vertex $v \in V(G)$ is \emph{blocked} (by $I$) when $N_G(v) \cap I \neq \emptyset$.

\section{Equivalence of the $k$-Jump Rule ($k \geq 3$)}
\label{sec:Simulate}
The goal of this section is to prove Theorem~\ref{thm:SimulateTJ}.
\begin{rtheorem}{thm:SimulateTJ}
Let $G$ be a connected graph, and $k \geq 3$ be an arbitrary integer.
An instance $(G, I_s, I_t)$ is reconfigurable under the $k$-Jump rule if and only if $(G, I_s,I_t)$ is reconfigurable under the $D(G)$-Jump rule. 
\end{rtheorem}
When independent sets $I_1$ and $I_2$ of $G$ satisfy $I_1 \transition{k} I_2$, they also satisfy $I_1 \transition{k'} I_2$ for any $k'>k$.
The converse also holds as the following lemma shows.

\begin{lemma}
Let $k' > k \geq 3$.
If $I_1 \transition{k'} I_2$ holds for independent sets $I_1$, $I_2$ of $G$, then $I_1 \transition{k} I_2$ holds.
\label{lemma8}
\end{lemma}
\begin{proof}

Without loss of generality, we consider only the case of $I_1 \neighbor{k'} I_2$ (that is, only a single token moves in the transition from $I_1$ to $I_2$).
Let $t$ be the token which moves from vertex $u_0$ to $u_h$ in transition from $I_1$ to $I_2$, and $P = u_0, u_1, \dots, u_h$ be a shortest path from $u_0$ to $u_h$.
The proof is by induction on $h$.
(Basis) $h \leq k$: the lemma obviously holds.
(Inductive Step) Assuming that the lemma holds for any $h' < h$, consider the case of $h$.
When $u_{h - k+1}$ is not blocked and has no token, then $t$ can move to $u_\ell$ via $u_{h - k+1}$ by induction assumption (Fig.~\ref{fig1}(a)) since $\dist_{G}(u_0, u_{h - k + 1}) \leq h'$ and $\dist_G(u_{h - k+1}, u_h) \leq k$ hold. 
Thus, $I_1 \transition{k} I_2$ holds.
When $u_{h - k + 1}$ has a token or is blocked, then a vertex $u' \in \{u_{h - k + 1}\} \cup N_G(u_{h - k + 1})$ has a token, say $t'$.
From $\dist_G(u_0, u') < h$ and $\dist_G(u', u_h) \leq k$, $t'$ can move to $u_h$ and then $t$ can move to $u'$ by induction assumption (Fig.~\ref{fig1}(b) and \ref{fig1}(c)). Thus $I_1 \transition{k} I_2$ holds.
\qed
\end{proof}

\begin{figure}[t]
  \begin{minipage}[b]{0.32\linewidth}
    \centering
    \includegraphics[keepaspectratio, scale=0.5, width =\linewidth]{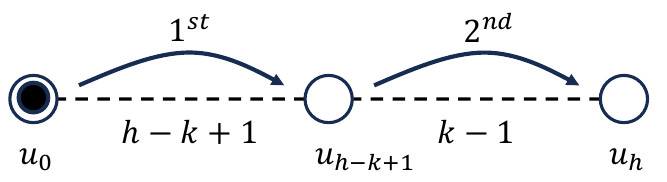}
    \subcaption[a]{}
  \end{minipage}
  \begin{minipage}[b]{0.32\linewidth}
    \centering
    \includegraphics[keepaspectratio, scale=0.5, width =\linewidth]{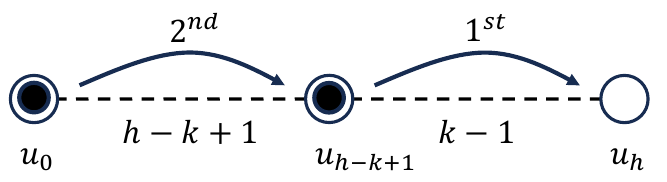}
    \subcaption[a]{}
  \end{minipage}
    \begin{minipage}[b]{0.32\linewidth}
    \centering
    \includegraphics[keepaspectratio, scale=0.5, width =\linewidth]{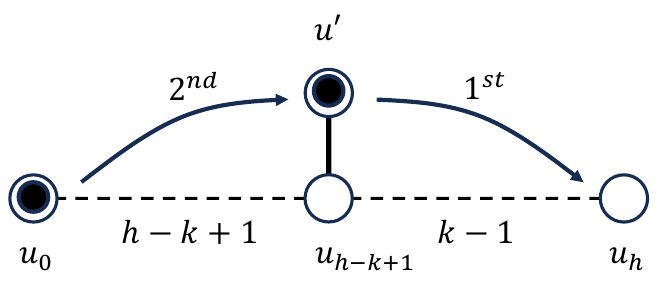}
    \subcaption[a]{}
  \end{minipage}
  \caption{
 Illustration for proof of Lemma \ref{lemma8}.
 Tokens move along the arrows in the described order.
  (a) Case where $u_{h - k + 1}$ is not blocked and has no token, (b) Case where $u_{h - k + 1}$ has a token.  (c) Case where $u_{h - k + 1}$ is blocked.
  }
    \label{fig1}
\end{figure}

\section{Algorithm for $2$-ISReconf on Split Graphs}
\label{sec:2-ISReconf}
In this section, we give a polynomial-time algorithm to solve {$2$-\Isreconf} for split graphs, and prove Theorem~\ref{theorem2}.
\begin{rtheorem}{theorem2}
There exists a polynomial-time algorithm for \Isreconf\ under the $2$-Jump rule for split graphs.
\end{rtheorem}
\subsection{Split Graphs and Fundamental Properties}
In this section, we present fundamental properties use in the algorithm.
A \emph{split graph} is a graph such that its vertex set can be partitioned into a clique and an independent set. 
Then, a split graph $G$ can be seen as the graph consisting of a complete graph $G^A = (V^A, E^A)$ and a bipartite graph $G^B=(V^B, U^B, E^B)$ such that $V(G) = V^A \cup U^B$, $V^B \subseteq V^A$, $U^B \cap V^A = \emptyset$, and $E(G) = E^A \cup E^B$. (See \figurename~\ref{fig2}).
Since one can identify the vertex sets $V^A$ and $U^B$ for a given split
graph $G$ in polynomial time, the following argument assumes that the information on those sets are available.
For simplicity, we assume that no isolated vertex exists in $G^B$.
Each connected component of the bipartite Graph $G^B$ is called a \emph{cluster} .
In the following, let the cluster set of the split graph $G$ be $\mathcal{C} = \{C_0,C_1, \dots, C_{m-1}\}$ with $C_i = (U^{\C}_{i}, V^{\C}_{i}, E^{\C}_{i})$ ($U^{\C}_i \subseteq U^B$, $V^{\C}_i \subseteq V^B$, and $E_i$ is the set of edges induced from $E^B$).
For convenience, when there exists a vertex set $V' \subseteq V^A$ not having any neighbor in $U^B$, we treat $(\emptyset, V', \emptyset)$ as a bipartite graph cluster and is included in $\mathcal{C}$.

\begin{figure}[t]
  \begin{minipage}[b]{0.3\linewidth}
    \centering
    \includegraphics[keepaspectratio, scale=0.25]{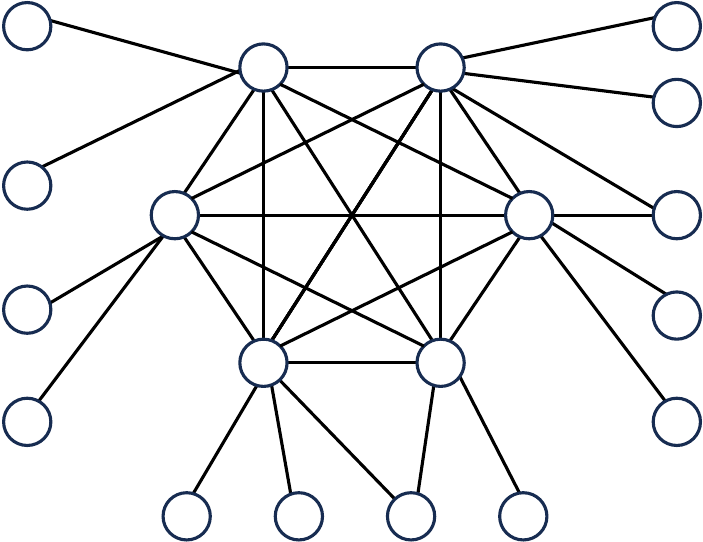}
    \subcaption[a]{}
  \end{minipage}
  \begin{minipage}[b]{0.3\linewidth}
    \centering
    \includegraphics[keepaspectratio, scale=0.25]{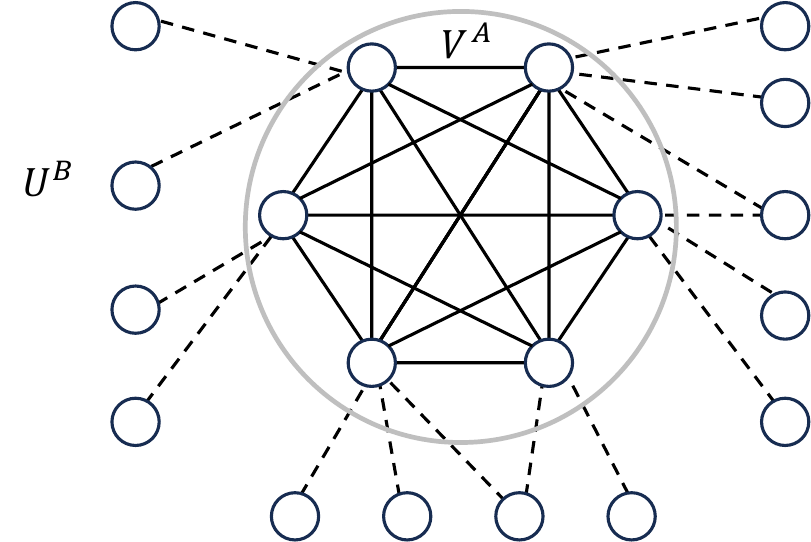}
    \subcaption[b]{}
  \end{minipage}
    \begin{minipage}[b]{0.3\linewidth}
    \centering
    \includegraphics[keepaspectratio, scale=0.25]{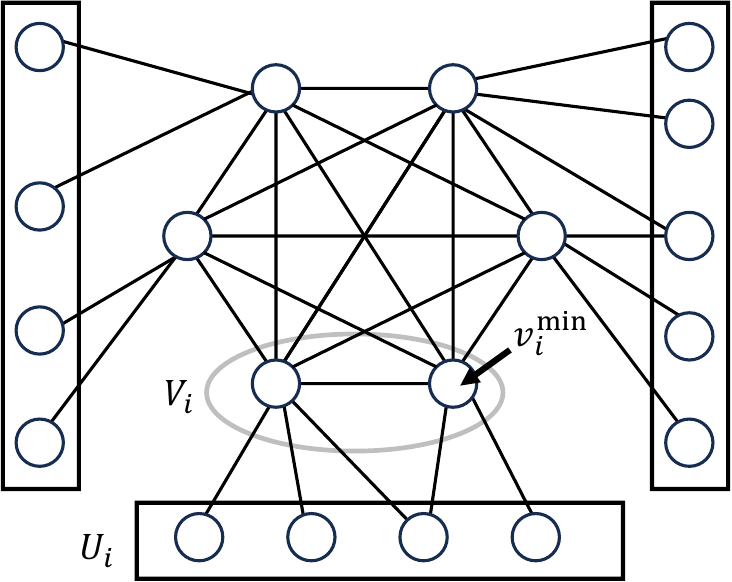}
    \subcaption[c]{}
  \end{minipage}
  \caption{
  (a) An example of a split graph.
  This graph consists of a complete subgraph composed of the solid lines in (b) and a bipartite graph composed of the solid lines in (c).
  The vertices in the complete subgraph in (b) constitute $V^A$, and the vertices in the bipartite graph contained in the squares in (c) constitute $U^B$.
  Each cluster $C_i$ is represented by a rectangle containing $U_i$ and its neighbor vertices.
  }
    \label{fig2}
\end{figure}

In the following, we exclude the obvious case of $|I_s| = |I_t| > |U^B|$.  In this case, the size of the maximum independent set of $G$ is at most $|U^B|+1$: tokens are placed on all vertices in $U^B$ in both $I_s$ and $I_t$ and another token is placed on a node in $V^A$ that is possibly different in $I_s$ and $I_t$.  Thus, $I_s \transition{2} I_t$ always holds.
In the case of $|I_s|=|I_t| \leq |U^B|$, we assume that $I_s$ and $I_t$ contain no vertex in $V^A$, which does not lose generality from the following reason.
Note that the diameter of the split graph is at most 3, and the distance from the vertex in $V^A$ to the vertex in $U^B$ is always at most 2.
When $I_s$ contains a node $v$ in $V^A$, we can obtain an independent set $I'_s$ by moving the token on $v$ to an arbitrary empty vertex in $U^B$. Similarly, we can obtain $I'_t$ from $I_t$.  It is clear that $I_s \transition{2} I_t$ holds if and only if $I'_s \transition{2} I'_t$ holds.
In the following, an independent set $I$ such that $I \cap V^A = \emptyset$ is called a \emph{typical} independent set.

\subsection{Token Distribution}

In this section, we state a lemma and its corollary. 
The following lemma shows that tokens can be freely moved within each cluster (independent of token placement outside the cluster).
\begin{lemma}
\label{lma:innercluster}
Let $I_1, I_2$ ($I_1 \neq I_2$) be typical independent sets of G with the same size (i.e., $|I_1|=|I_2|$) such that their token placements are different only in some cluster $C_i=(U^{\C}_i, V^{\C}_i, E^{\C}_i)$, that is, $|I_1 \cap U^{\C}_i| = |I_2 \cap U^{\C}_i|$ and $I_1 \setminus U^{\C}_i = I_2 \setminus U^{\C}_i$ are satisfied.
Then, $I_1 \transition{2} I_2$ holds.
\end{lemma}
\begin{proof}
It is sufficient to show that the lemma holds for $I_1$ and $I_2$ such that $|I_1 \symdiff I_2| = 2$.    
Let $u$ (resp. $v$) be a vertex in $I_1 \setminus I_2$ (resp. $I_2 \setminus I_1$) and $P=u_0 (=u), v_1, u_1, v_1, \dots, u_{y} (=v)$ ($u_j \in U^{\C}_i$ and $v_j \in V^{\C}_i$) be a path in $C_i$ from $u$ to $v$.
Also, let $j_0, j_1, \dots, j_{h - 1}$ ($j_0 = 0$ and $j_x < j_{x+1}$) be the index sequence of nodes in $P \cap I_1\ (\subseteq U^{\C}_i)$ and let $j_h = y$.
Notice that $u_y$ is empty in $I_1$, and no node in $U^{\C}_i$ is blocked as long as no token exists in $V^{\C}_i$.
Thus, we can reconfigure the token placement from $I_1$ to $I_2$ by moving the token on $u_{j_x}$ to $u_{j_{x+1}}$ in descending order of $x\ (0 \le x \le y-1)$.
\qed
\end{proof}
%
Given a typical independent set $I$, we define the \emph{distribution} of $I$ as vector $(|I \cap U^{\C}_i|)_{0 \leq i \leq m - 1}$.
The following corollary is derived from Lemma~\ref{lma:innercluster}.
\begin{corollary}
\label{corol:samedistribution}
If typical independent sets $I$ and $I'$ have the same distribution, then $I \transition{2} I'$.
\end{corollary}

\subsection{Cluster Types}
In this section, we introduce the crucial notion of cluster types.
For each $i$, let $v^{\min}_i$ be any vertex in  $V^{\C}_i$ with the minimum degree in $C_i$ and let $N_i = N_{C_i}(v^{\min}_i)$.
In the following, we assume without loss of generality that $|N_0| \leq |N_1| \leq \dots \leq |N_m-1|$ (recall that m is the number of
    clusters of $G$).
Given an independent set $I$, we call $f_i(I) = |N_i \cap I|$ the \emph{occupancy} of cluster $C_i$ on $I$.
By definition, the occupancy of a cluster $C_i$ satisfying $U^{\C}_i = \emptyset$ is 0.
For a typical independent set $I$ of $G$, let $I^{\ast}$ be the typical independent set with the same distribution as $I$ such that $f_i(I^\ast)$ of each cluster $C_i$ is minimum among all typical independent sets with the same distribution as $I$.
We now consider the following classification of clusters.
\begin{definition}
\label{def:clastertype}
Given a typical independent set $I$, a cluster $C_i$ is called \emph{Free} if $f_i(I^{\ast}) = 0$, \emph{Pseudo-Free} if $f_i(I^{\ast}) = 1$, and \emph{Bounded} otherwise.
This is referred to as the \emph{type} of the cluster $C_i$, and only depends on the distribution of $C_i$.
\end{definition}

The properties of each cluster type for a typical independent set $I$ are intuitively described next.
\begin{description}
\item[\textbf{Free cluster}] If there exists an $i$ such that $C_i$ is Free,
    then $N_{G}(v^{\min}_i) \cap I^{\ast} = \emptyset$ by definition (recall that $I^{\ast}$ contains no vertex in $V^A$).
Also, since the distance between any vertex in $G$ and $v^{\min}_i \in V^A$ is at most 2, after transforming $I$ to $I^{\ast}$ by Corollary \ref{corol:samedistribution}, any token in any cluster $C_j$ can be moved via $v^{\min}_i$ to any vertex in any distinct cluster $C_h$.
This move is possible even if $h = i$, which possibly makes $C_i$ become Pseudo-Free.
\item[\textbf{Pseudo-Free cluster}] If $C_i$ is Pseudo-Free, after transforming $I$ to $I^{\ast}$, the token in $N_i \cap I^{\ast}$ can be moved via $v^{\min}_i$ to any cluster vertex, which makes $C_i$ become Free.
\item[\textbf{Bounded cluster}] If $C_i$ is Bounded, tokens can move into $C_i$ from other clusters (and vice versa) only if there exists a Free Cluster.
\end{description}
\begin{figure}[t]
  \begin{minipage}[b]{0.32\linewidth}
    \centering
    \includegraphics[keepaspectratio, scale=0.25]{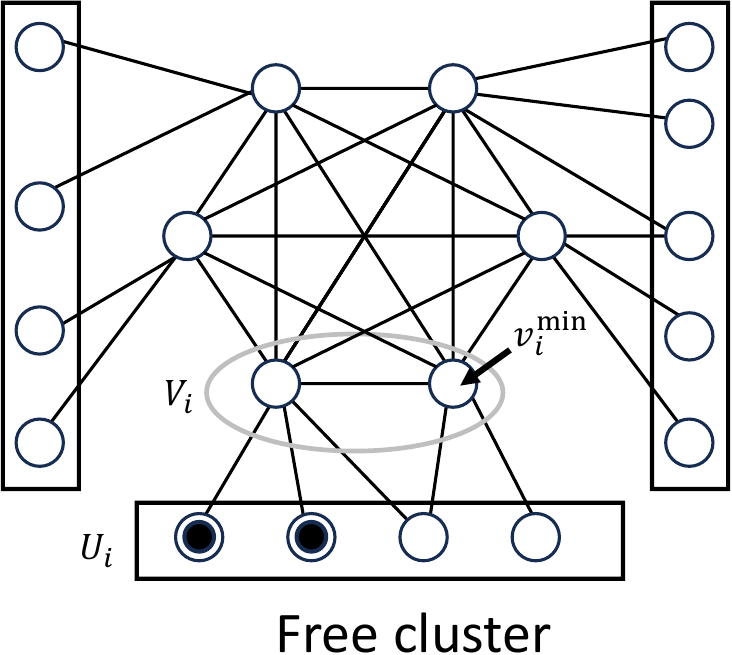}
    \subcaption[a]{}
  \end{minipage}
  \begin{minipage}[b]{0.32\linewidth}
    \centering
    \includegraphics[keepaspectratio, scale=0.25]{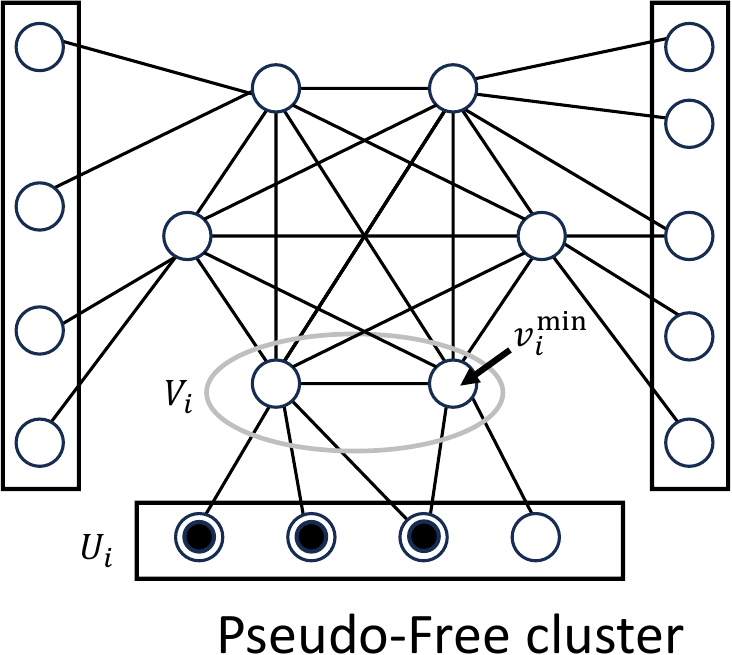}
    \subcaption[b]{}
  \end{minipage}
    \begin{minipage}[b]{0.32\linewidth}
    \centering
    \includegraphics[keepaspectratio, scale=0.25]{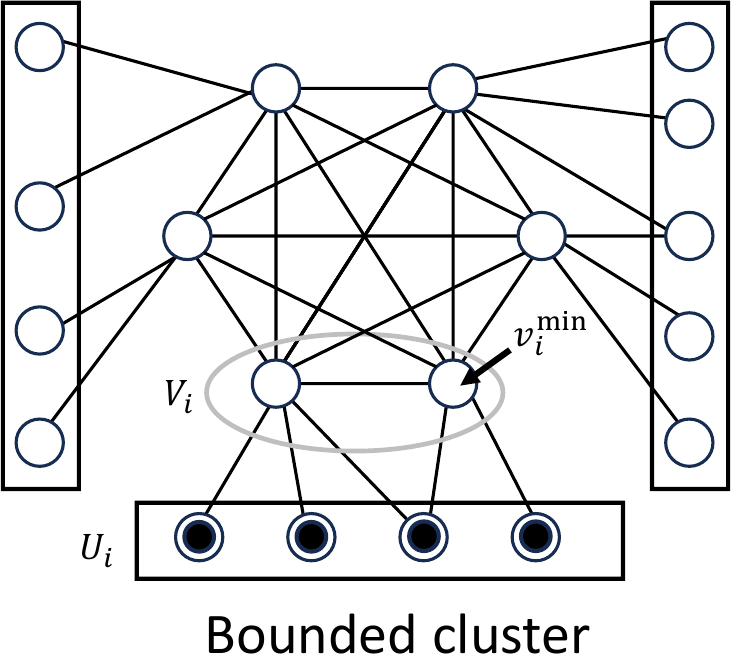}
    \subcaption[c]{}
  \end{minipage}
  \caption{
  An example of the cluster types. The cluster type of $C_i$ is (a) a Free Cluster, (b) a Pseudo-Free Cluster, and  (c) a Bounded Cluster.
  }
    \label{fig:cluster_types}
\end{figure}
Examples of cluster types are shown in Figure~\ref{fig:cluster_types}.
We say that cluster $C_i$ is \emph{full} in a typical independent set $I$ if $I \cap U^{\C}_i = U^{\C}_i$.
By definition, no Free Cluster is full and a Pseudo-Free Cluster $C_i$ is full only if $|N_i| = 1$.
Also, for any two independent sets $I$ and $I'$ with the same distribution, if the type of $C_i$ is $X$ for $I$, then its type is also $X$ for $I'$.
Similarly, if $C_i$ is full for $I$, then $C_i$ is full for $I'$.
In the following, let $\mathcal{F}(I) \subseteq \mathcal{C}$ be the set of Free Clusters for $I$.
We show three lemmas.
\begin{lemma}
\label{lma:block}
 If a cluster $C_h\in\mathcal{C}$ is Pseudo-Free or Bounded for a typical independent set $I$, then the vertices of $V^{\C}_h$ are all blocked by $I$.
\end{lemma}

\begin{proof}
We prove the lemma by contradiction.
If a vertex $v \in V^{\C}_h$ is unblocked, then $N_{C_h}(v) \cap I = \emptyset$, which implies  
$|I \cap U^{\C}_h| \leq |U^{\C}_h| - |N_{C_h}(v)|$.
Since $|N_h| \leq |N_{C_h}(v)|$ by definition, $I^{\ast} \cap N_i = \emptyset$ and thus $C_h$ is Free, which is a contradiction.
\qed
\end{proof}

\begin{lemma}
\label{lma:notfree}
Let $I$ be a typical independent set satisfying one of the following conditions. 
\begin{description}
\item[\textbf{(C1)}] All clusters are Bounded for $I$.
\item[\textbf{(C2)}] For any Pseudo-Free Cluster $C_i$ for $I$, all clusters in $\mathcal{C} \setminus \{C_i\}$ are full.
\end{description}
Then, any typical independent set $I'$ such that $I \transition{2} I'$ has the same distributions as $I$.
\end{lemma}
\begin{proof}
We show that any transition sequence starting from $I$ cannot change the distribution at typical independent sets.
Considering any transition $I \neighbor{2} \hat{I}$, let $t$ be the token moving in this transition, and $C_i$ be the cluster where $t$ is placed in $I$.
If $I$ satisfies condition \textbf{(C1)}, then each cluster is Pseudo-Free or Bounded for $I \setminus \{t\}$, and thus all vertices in $V^A$ are blocked from Lemma~\ref{lma:block}.
Since the distance between any vertex in $U^{\C}_i$ and any vertex in $U^{\C}_h$ ($h \neq i$) is 3, 
$t$ can move to only a vertex in $U^{\C}_i$, which preserves the distribution.
Consider the case that $I$ satisfies condition \textbf{(C2)}.  Condition \textbf{(C2)} implies that no cluster is Free and thus one cluster $C_j$ is Bounded or Pseudo-Free.  If $C_j$ is Bounded, $t$ can move to only a vertex in $U^{\C}_j$ since all clusters other than $C_i$ are full, which preserves the distribution.  If $C_j$ is Pseudo-Free, $t$ can move to only a vertex in $V^A \cup U^{\C}_j$.  When $t$ moves to a vertex in $U^{\C}_j$, the distribution remains unchanged.  When $t$ moves to a vertex in $v^A$, the vertex in $V^A$ is $v^{\min}_j$ since only $v^{\min}_j$ is not blocked in $I \setminus \{t\}$.  The token on $v^{\min}_i$ can move to only a vertex in $V^A \cup U^{\C}_j$.  By repeating the argument, we can show that $I$ and $I'$ have the same distribution.
\qed
\end{proof}

\begin{lemma}
\label{lma:commonfree}
If $\mathcal{F}(I_s) \cap \mathcal{F}(I_t) \neq \emptyset$, then $I_s \transition{2} I_t$.
\end{lemma}
\begin{proof}
Let $C_i$ be any cluster in $\mathcal{F}(I_s) \cap \mathcal{F}(I_t)$.
We can assume $|I_s \cap U^{\C}_i| \geq |I_t \cap U^{\C}_i|$ by the symmetry of the relation $\transition{2}$, and $I_s = I^{\ast}_s$ and $I_t = I^{\ast}_t$ by Lemma~\ref{lma:innercluster}.
The proof is by induction on the size of $|I_s \symdiff I_t|$.

\noindent
(Basis) When $|I_s \symdiff I_t|=0$ (or $I_s=I_t$), then it is obvious that $I_s \transition{2} I_t$.

\noindent
(Inductive Step) Assuming the lemma holds for any $I_s$ and $I_t$ with $|I_s \symdiff I_t| \leq k$ $(k \geq 0)$, we prove the lemma for $|I_s \symdiff I_t|= k + 2$.
Let $u \in I_s \setminus I_t$ and $u' \in I_t \setminus I_s$.
Because both $I_s$ and $I_t$ are typical,  $v^{\min}_i$ is not blocked.
Also, the distance between $v^{\min}_i$ and any vertex is at most 2.
Thus, the token on $u$ can move to $u'$ via $v^{\min}_i$.
Let  $I'_s$ be the independent set after the move, then $I'_s$ is typical and $|I'_s \symdiff I_t| = k$ holds.
Unless $u \not\in U^{\C}_i$ and $u' \in U^{\C}_i$, $C_i$ remains Free for $I'_s$ and thus $I'_s \transition{2} I_t$ by inductive assumption, which implies $I_s \transition{2} I_t$.
In the case of $u \not\in U^{\C}_i$ and $u' \in U^{\C}_i$, let $I'_t$ be an independent set obtained by moving a token on $u'$ to $u$.
By the same argument, $I_s \transition{2} I_t$ is shown from $I_s \transition{2} I'_t$.
\qed
\end{proof}
%

\subsection{Technical Ideas for the Algorithm}
In this section, we explain some technical ideas for the algorithm (fully described in Section~\ref{lma:putting_all}). We first consider the following three cases.
%
\begin{enumerate}
\item For $I_s$, all clusters are Bounded.
\item For $I_s$, there exists no Free Cluster, and one or more clusters are Pseudo-Free.
\item For $I_s$, a Free Cluster exists.
\end{enumerate}

For case 1, by Lemma~\ref{lma:notfree}, $I_s \transition{2} I_t$ holds if and only if $I_s$ and $I_t$ have the same distribution.
For case 2 satisfying condition \textbf{(C2)} of Lemma~\ref{lma:notfree}, $I_s \transition{2} I_t$ holds if and only if $I_s$ and $I_t$ have the same distribution.
Thus, in the above cases, whether $I_s \transition{2} I_t$ holds or not can be determined in polynomial time.
For case 2 not satisfying condition \textbf{(C2)}, we can make a Pseudo-Free Cluster $C_i$ Free by moving one token from $C_i$ to a non-full cluster, which leads us to Case 3.
So the remaining case we need to consider is case 3.
By a similar argument for $I_t$, the only case we need to consider is the one where a Free Cluster exists in $I_t$.
So we consider only the case where $I_s$ and $I_t$ has a Free Cluster respectively.

When $I_s$ and $I_t$ have a common Free Cluster, Lemma~\ref{lma:commonfree} guarantees $I_s \transition{2} I_t$.  Otherwise, $I_s \transition{2} I_t$ holds if there exist $I'_s$ and $I'_t$ such that $I_s \transition{2} I'_s$, $I_t \transition{2} I'_t$ and $\mathcal{F}(I'_s) \cap \mathcal{F}(I'_t) \neq \emptyset$.
The following lemma holds.
\begin{lemma}
\label{lma:sufficientcondition}
Let $C_i$ be any Free Cluster for a typical independent set $I$.
Let $C_j\ (j \ne i)$ be any cluster for $I$ satisfying $|N_i| \geq k$ and $|U^B| \geq |I|+|N_i| + |N_j| -k$ for some $k \in \{0, 1, 2\}$, then there exists a typical independent set $I'$ such that $I \transition{2} I'$ and $C_j$ is Free for $I'$.  Furthermore, $I'$ can be found in polynomial time.
\end{lemma}
\begin{proof}
Without loss of generality, we assume $I = I^{\ast}$ by Lemma~\ref{lma:innercluster}.
Let $N^\prime=U^B \setminus (N_i \cup N_j)$.
Because $N^\prime$,$N_i$,$N_j$ are mutually disjoint, $|I|=|I \cap N^\prime|+|I \cap N_i|+|I \cap N_j|$.
Since $C_i$ is Free for $I$, $|I \cap N_i|=0$ holds and thus $|I|=|I \cap N^\prime|+|I \cap N_j|$.
\begin{align*}
& |U^B| \geq |I|+|N_i| + |N_j| - k \\
&\Leftrightarrow |U^B| - |N_i| - |N_j| \geq |I| - k\\
&\Leftrightarrow |N^\prime| \geq |I\cap N^\prime| + |I\cap N_j| - k\\
&\Leftrightarrow |N^\prime\setminus I| \geq |I \cap N_j| -k
\end{align*}
The last inequality implies that there exist at least $|I \cap N_j| - k$ empty vertices in $N^\prime$.
Using the Free Cluster property of $C_i$, $|I \cap N_j| - k$ tokens on $I \cap N_j$ can be moved to vertices in $N'$ via $v^{\min}_i$, which leaves $k$ tokens in $N_j$ ($0 \leq k \leq 2$).
Let $\hat{I}$ be the typical independent set after the tokens move.
When $k=0$,  $\hat{I}$ is $I^\prime$.
When $k=1$, $I^\prime$ is obtained from $\hat{I}$ by moving the remaining token in $N_j$ to a vertex in $N_i$ via $v^{\min}_i$.
When $k=2$, one of the remaining tokens can be moved to a vertex in $N_i$.  For the resultant independent set, $C_j$ is Pseudo-Free.  Thus $I^\prime$ can be obtained by moving the last token in $N_j$ to a vertex in $N_i$ via $v_j^{min}$ (from $k=2$, $N_i$ contains an empty vertex).  It is clear that $I^\prime$ can be found in polynomial time.
\qed
\end{proof}
When a cluster $C_j$ satisfies the condition of Lemma~\ref{lma:sufficientcondition}, any cluster $C_{j^\prime}\ (j' \leq j)$ also satisfies the condition because of $|N_{j'}| \leq |N_j|$.
Similarly, when a Free Cluster $C_i$ for $I$ satisfies the condition for some $j$, any Free Cluster $C_i^\prime\ (i^\prime < i)$ also satisfies the condition for $j$.
Thus, without loss of generality, the lemma can assume that $i$ is the smallest such that $C_i$ is Free for $I$ and $j = 0$.
By combining with Lemma~\ref{lma:commonfree}, the following corollary
is derived.
In the corollary, $i(I)$ denotes the minimum $i$ such that $C_i$ is Free for $I$.
\begin{corollary}
\label{corol:sufficient}
Let $I_s$ and $I_t$ be any typical independent sets having a Free Cluster respectively and $i(I_s) \neq i(I_t)$ holds.  If both of the following two conditions hold, then $I_s \transition{2} I_t$.
\begin{itemize}
\item For some $k_1 \in \{0, 1, 2\}$, $|N_{i(I_s)}| \geq k_1$ and $|U^B| \geq |I_s|+|N_{i(I_s)}| + |N_0| - k_1$,
\item For some $k_2 \in \{0, 1, 2\}$, $|N_{i(I_t)}| \geq k_2$ and $|U^B| \geq |I_t|+|N_{i(I_t)}| + |N_0| - k_2$
\end{itemize}
\end{corollary}

This corollary gives us a sufficient condition for $I_s \transition{2} I_t$, but in fact, the following lemma shows that it is also a necessary condition.
\begin{lemma}
\label{lma:necessarycondition}
Let $I_s$ and $I_t$ be any typical independent sets having a Free Cluster respectively and $i(I_s) \neq i(I_t)$ holds. Both of the following two conditions hold if $I_s \transition{2} I_t$.
\begin{itemize}
\item For some $k_1 \in \{0, 1, 2\}$, $|N_{i(I_s)}| \geq k_1$ and $|U^B| \geq |I_s|+|N_{i(I_s)}| + |N_0| - k_1$,
\item For some $k_2 \in \{0, 1, 2\}$, $|N_{i(I_t)}| \geq k_2$ and $|U^B| \geq |I_t|+|N_{i(I_t)}| + |N_0| - k_2$.
\end{itemize}
\end{lemma}
\begin{proof}
We can assume, without loss of generality, $I_s = I^{\ast}_t$, $I_t = I^{\ast}_t$ by Lemma~\ref{lma:innercluster}.
For contradiction, assume that the conditions of the lemma are not satisfied.
By symmetry, without loss of generality, we assume that $I_s$ does not satisfy the condition. 
Also, we denote $i = i(I_s)$ for short.
When $|N_i| = 0$, $C_i$ is Free for any typical independent set $I$, which contradicts to $i(I_s) \neq i(I_t)$.
When $|N_i| = 1$,  by assumption, $|U^B| < |I_s|+|N_i| + |N_0| - 1$ holds.
It follows from $|N_0| \leq |N_i| = 1$ (by definition) that $|U^B \setminus I_s| = |U^B| - |I_s| < 1$ holds.
Since $C_i$ is Free and satisfies $N_i \subseteq U^B \setminus I_s$, $|N_i| \leq |U^B \setminus I_s| < 1$ holds, which is a contradiction.

We consider the case of $|N_i| \geq 2$.
Let $I$ be any typical independent set such that $I_s \transition{2} I$ holds and $C_i$ is Free for $I$.
Since $C_i$ does not satisfy the condition for $k_1=2$, $|U^B| < |I_s|+|N_i| + |N_0| - 2 = |I| + |N_i| + |N_0| - 2$ holds.
Since $C_i$ is Free, $I \cap N_i = \emptyset$ and $|I \cup N_i| = |I| + |N_i|$ hold, which derives $|U^B \setminus (I \cup N_i)| < |N_0| - 2$.
For any cluster $C_j\ (j \ne i)$, $|N_j \setminus I| < |N_0| - 2$ holds from $N_j \subseteq U^{B} \setminus N_i$.

\begin{align*}
& |N_j| = |N_j \cap I| + |N_j \setminus I| \geq |N_0|\ (\because\ |N_j|\ge |N_0|) \\
&\Leftrightarrow |N_j \cap I| \geq |N_0| - |N_j \setminus I| \\
&\Leftrightarrow |N_j \cap I| > 2 
\end{align*}
On the other hand, from $I_s \transition{2} I_t$, there exists a reconfiguration sequence $I_0(=I_s), I_1, \ldots, I_\ell (= I_t)$.
Let $h$ be the maximum index such that $C_i$ is Free for $I_h$.
Since $C_i$ is not Free for $I_t$, $h < \ell$ hols.
Also, since $C_i$ is not Free for $I_{h+1}$, a token moves to a vertex in $N_i$ in the transition from $I_h$ to $I_{h+1}$.
It follows from the above inequality that $I_h \cap C_j > 2$ holds for any cluster $C_j\ (j \ne i)$ and thus $I_{h+1} \cap C_j \geq 2$ holds.
Similarly, $C_i$ is Pseudo-Free for $I_{h+1}$.
These imply only $C_i$ is Pseudo-Free and all other clusters are Bounded for $I_{h+1}$.
From the definition of $h$, $C_i$ is not Free for $I_{h'}$ if $h'> h$.  This requires one cluster other than $C_i$ need to become Free in reconfiguration sequence $I_{h+1}, I_{h+2}, \ldots, I_\ell (= I_t)$ without making $C_i$ Free.
However, all clusters other than $C_i$ are Bounded for $I_{h+1}$, so such a reconfiguration sequence is impossible because of the properties of Bounded Clusters.
\qed
\end{proof}

\subsection{Putting All Together: the Algorithm}
\label{lma:putting_all}

In summary, we show a polynomial-time decision algorithm for {$2$-\Isreconf} on split graphs.  We assume, without loss of generality, that given $I_s$ and $I_t$ are typical independent sets.

For a given split graph $G = (V^A \cup U^B, E^A \cup E^B)$, we first obtain clusters $C_0,\ldots,C_{m-1}$ by deleting all edges in $E^A$ (or edges in the complete subgraph). 
Then, we obtain the following elements for each cluster $C_i$.
\begin{itemize}
    \item $v^{\min}_i$: the vertex with the minimum degree in $V^{\C}_i$.
    \item $|N_i|$: the degree of $v^{\min}_i$ in $C_i$.
    \item $|I_s \cap U^{\C}_i|$ and $|I_t \cap U^{\C}_i|$ for each $i\ (0 \le i \le m-1)$: the numbers of tokes in cluster $C_i$ for $I_s$ and $I_t$ respectively.
\end{itemize}

We then classify the clusters for $I_s$ and $I_t$ into Free Clusters, Pseudo-Free Clusters, and Bounded Clusters: $C_i$ is Free for $I (\in \{I_s, I_t\})$ if $|U^{\C}_i|-|I \cap U^{\C}_i| \ge |N_i|$, Pseudo-Free if $|U^{\C}_i|-|I \cap U^{\C}_i| = |N_i|-1$, or Bounded otherwise.
Similarly, we determine whether $C_i$ is full or not for $I$.

After the classification, we check whether $I_s \transition{2} I_t$ or not.  If all the clusters are Bounded, or only a single cluster is Pseudo-Free and all other clusters are full, then we can determine, following Lemma~\ref{lma:notfree}, whether $I_s \transition{2} I_t$ or not by checking whether they have the same distribution or not.
If $I_s$ and $I_t$ have a common Free Cluster, then we can determine,  following Lemma~\ref{lma:commonfree}, that $I_s \transition{2} I_t$ holds.
Finally, if no cluster is Free for at least one of $I_s$ and $I_t$, then we can determine, following Lemma~\ref{lma:sufficientcondition} and Lemma~\ref{lma:necessarycondition}, whether $I_s \transition{2} I_t$ or not by checking whether both the following condition are satisfied or not. 
\begin{itemize}
    \item $N_i$=1 and $|U^B| \geq |I_s|+|N_i| + |N_0| -1$, or $|N_i| > 1$ and $|U^B| \geq |I_s|+|N_i| + |N_0| -2$.
    \item $N_{i'}=1$ and $|U^B| \geq |I_t|+|N_{i'}| + |N_0| -1$, or $|N_{i'}| > 1$ and $|U^B| \geq |I_t|+|N_{i'}| + |N_0| -2$.
\end{itemize}
It is obvious that the procedure described above can be executed in polynomial time.

\section{NP-completeness of {\ShortestunderkIsreconf}}
\label{sec:NP-completeness}
In this section, we prove Theorem~\ref{theorem3}.
\begin{rtheorem}{theorem3}
Let $k \geq 3$ be any integer. 
Then, there exists a graph class $\mathcal{G}_k$ such that $\mathcal{G}_k$ is a subclass of chordal graphs of diameter at most $2k + 1$ and \ShortestIsreconf\ under the $k$-Jump rule is NP-complete for $\mathcal{G}_k$.
\end{rtheorem}

To prove the theorem, we give a polynomial-time reduction from the E3-SAT problem.
The E3-SAT problem is a special case of the SAT problem, where each clause contains exactly three literals. 
We reduce any instance $\Phi$ of E3-SAT to the 
instance $\Phi' = (G,I_s,I_t)$ of {\ShortestunderkIsreconf} whose shortest reconfiguration sequence has a length at most $2(m+n)$ 
if and only if $\Phi$ is satisfiable, where $m$ and $n$ is the number of clauses and variables in $\Phi$.

\subsection{Gadget Construction}
\label{sec:gadget_construct}

Consider any instance $\Phi$ of E3-SAT consisting of $m$ clauses $c_0, c_1, \dots,c_{m-1}$ and $n$ variables $x_0, x_1, \dots, x_{n-1}$. 
We construct the \emph{clause gadget} $C_i$ for each clause $c_i$ in $\Phi$, 
and construct the \emph{variable gadget} $L_j$ for each variable $x_j$ in $\Phi$. 

\paragraph*{Clause Gadget}
We define the clause gadget $C_i$.
The gadget $C_i$ under the $k$-Jump rule is defined as follows (see Fig. \ref{fig3}):
\begin{itemize}
\item Create a path $P=(v_{0},v_{1},\dots,v_{2k-1},v_{2k})$, and define aliases $s$, $k_1$, and $t$ as $s=v_0$, $k_1=v_k$, and $t=v_{2k}$. 
\item Add two vertices $k_0$ and $k_2$, and add four edges $\{k_0, v_{k-1}\}$, $\{k_0, v_{k+1}\}$, $\{k_2, v_{k-1}\}$ and $\{k_2, v_{k+1}\}$.
\end{itemize}
For any vertex $v$ in $C_i$, $v^i$ represents the vertex $v$ in the clause gadget $C_i$. 
Let $K=\bigcup_{i=0}^{m-1}\{k^i_0, k^i_1, k^i_2\}$. We further augment some edges crossing different clause gadgets.
\begin{itemize}
\item Connect any two vertices in $K$, i.e., $K$ forms a clique.
\end{itemize}

\begin{figure}[t]
\centering
\includegraphics[scale=0.5]{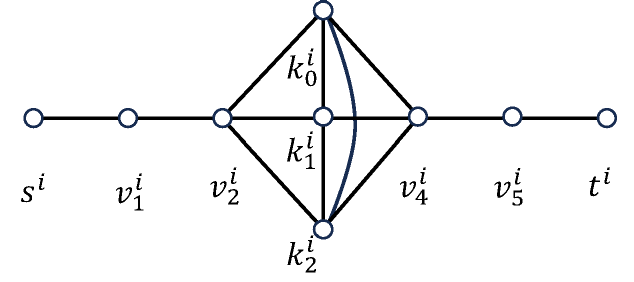}
\caption{Example of clause gadget $C_i$ when $k=3$. Note that edges of $\{k_0,k_1\}$, $\{k_1,k_2\}$ and $\{k_2,k_0\}$ are added in the last step of making $K$ a clique.} 
\label{fig3}
\end{figure}

\paragraph*{Variable Gadget}
The variable gadget $L_j$ under the $k$-Jump rule is constructed as follows (see also Fig. ~\ref{fig4}):
\begin{figure}[t]
\centering
\includegraphics[scale=0.5]{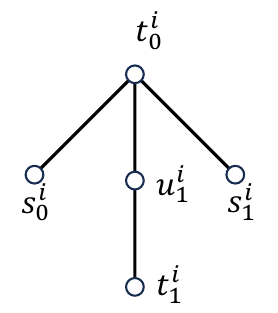}
\caption{Example of variable gadget $L_j$ when $k=3$.}
\label{fig4}
\end{figure}
\begin{itemize}
\item Create a path $P= (u_{0},u_{1},\dots,u_{k-2},u_{k-1})$. We give aliases $t_0$ and $t_1$ as $t_0=u_{0}, t_1=u_{k-1}$. 
\item Add two vertices $s_0, s_1$, and add two edges $\{s_0, t_0\}$ and $\{s_1, t_0\}$.
\end{itemize}
Similarly to the clause gadgets, for any vertex $v$ in $L_j$, $v^j$ represents the vertex $v$ in $L_j$.

\paragraph*{Whole Construction}
We obtain the whole graph $G$ by adding the edges connecting clause gadgets and variable gadgets, as follows:
\begin{itemize}
\item We perform the following process for each clause $c_i = (a \vee b \vee c)$. Let $L_{a}$ (resp. $L_{b}$, $L_{c}$) be the variable gadgets corresponding to $a$ (resp. $b$, $c$) and $\rho_i:\{a,b,c\}\rightarrow\{k^i_0, k^i_1, k^i_2\}$ be the function such that $\rho_i(a)=k^i_0$, $\rho_i(b)=k^i_1$, and $\rho_i(c)=k^i_2$. For all $\alpha\in\{a,b,c\}$.
If $\alpha$ is a positive literal, we add two edges $e_0=\{s^\alpha_0, \rho_i(\alpha)\}$ and $e_1=\{t^\alpha_0, \rho_i(\alpha)\}$.
Otherwise, we add two edges $e_0=\{s^\alpha_1, \rho_i(\alpha)\}$ and $e_1=\{t^\alpha_0, \rho_i(\alpha)\}$.
\end{itemize}
We finish the construction of $\Phi'$ by defining the initial independent set $I_s$ and the target independent set $I_t$ as follows:
\begin{itemize}
\item $I_s=\bigcup_{i=0}^{m-1}{v^i_0} \cup \bigcup_{j=0}^{n-1}{(s^j_0 \cup s^j_1)}$; and 
\item $I_t=\bigcup_{i=0}^{m-1}{v^i_{2k}} \cup \bigcup_{j=0}^{n-1}{(t^j_0 \cup t^j_1)}$.
\end{itemize}


\subsection{Proof of Theorem~\ref{theorem3}}

We define $\mathcal{G}_k$ as the family of the graphs constructed by the reduction from any E3-SAT instance explained above. First, we explain that any graph that is included in $\mathcal{G}_k$ is a chordal graph. If the induced graph by $v$ and $N(v)$ is clique, $v$ is called a \emph{simplicial vertex}. The \emph{perfect elimination ordering (PEO)} of $G$ is a vertex sequence $\pi=(p_0,\dots,p_{n-1})$ such that all $p_i$ are simplicial vertex in $G[V_i]$, where $G[V_i]$ is the induced graph by vertex set $\{p_i,p_{i+1},\dots,p_{n-1}\}$.
It is known that a graph $G$ has a PEO if and only if $G$ is a chordal graph~\cite{fulkerson1965incidence}.
The key technical lemmas are presented below:
\begin{lemma}
\label{lemma10}
For any $H \in \mathcal{G}_k$, $H$ is a chordal graph.
\end{lemma}
\begin{proof}
To prove that $H$ is a chordal graph, we show that $H$ has a PEO.
The number of vertices in $H$ is $n(2k+3)+m(k+2)$ because $n$ vertex gadgets and $m$ clause gadgets exist in $H$.
We consider the following vertex order $\pi=(p_0,\dots,p_{m(2k+3)+n(k+2)})$.
\begin{enumerate}
\item For any $i$ ($0\leq i \leq m-1$) and $h$ ($0 \leq h \leq k-1$), $p_{ik+h} = v^i_h$.
\item For any $i$ ($0\leq i \leq m-1$) and $h$ ($0 \leq h \leq k-1$), $p_{mk+ik+j} = v^i_{2k-h}$.
\item For any $j$ ($0\leq j \leq n-1$) and $h$ ($0 \leq h \leq k-2$), $p_{2mk+jk+h} = u^j_{k-h-1}$.
\item For any $i$ ($0\leq j \leq n-1$), $p_{2mk+n(k-1)+j}=s^j_0$.
\item For any $i$ ($0\leq j \leq n-1$), $p_{2mk+nk+j}=s^j_1$.
\item For any $j$ ($0\leq j \leq n-1$), $p_{2mk+n(k+1)+j}=t^j_0$.
\item For any $i$ ($0\leq i \leq m$) and $h$($0\leq h \leq 2$), $p_{2km+n(k+2)+3i+h}=k^i_h$.
\end{enumerate}
Except for steps 5 and 6, each vertex $p_i$ either has an adjacent vertex set that is a subset of clique $K$ or has degree 1 in the graph $H[V_i]$. In steps 5 and 6, for any $0\leq j \leq n-1, 0 \leq h \leq 1$, $s^{j}_{h}$ has only a subset of clique $K$ and $t^{j}_{0}$ as adjacent vertices in the graph $H[V_i]$. The adjacent vertices of $s^{j}_{h}$ that are included in $K$ are also included in adjacent vertices in $t^{j}_{0}$. So, it is easy to check that each vertex $p_i$ is simplicial vertex in $H[V_i]$. That is, $\pi$ is a PEO of $H$.
It implies that $H$ is a chordal graph.
\end{proof}

\begin{lemma}
\label{lemma12}
Let $G$ be the graph that are constructed by the reduction from an E3-SAT instance $\Phi$. 
The length of the solution of {\ShortestIsreconf} under the $k$-Jump rule for  instance $(G,I_s,I_t)$ is at most $2(m+n)$ if and only if  $\Phi$ is satisfiable.
\end{lemma}

We consider the proof of Lemma~\ref{lemma12}. First, we focus on the proof of the if part. \begin{lemma}
\label{lemma14}
Let $G$ be the graph constructed from an E3-SAT instance $\Phi$ by the reduction explained in Section~\ref{sec:gadget_construct}.
If the instance $\Phi$ is satisfiable, then there exists the reconfiguration sequence from the initial independent set $I_s$ to the target independent set $I_t$ such that the length of the sequence is at most~\footnote{Precisely, this is exactly $2(m + n)$. Since every token on a clause gadget has to jump twice or more ($2m$) and every token on a vertex gadget has to jump at least once ($2n$). So trivially no sequence with fewer moves is possible.} $2(m+n)$.
\end{lemma}
\begin{proof}
Since $\Phi$ is satisfiable, there is at least one assignment to $x_0,...,x_{n-1}$ such that it satisfies $\Phi$.
We consider fixing one assignment to $x_0,...,x_{n-1}$ that satisfies $\Phi$.
We perform the transition from $I_s$ to $I_t$ as follows:
\begin{description}
\item[\textbf{(M1)}] For each variable $x_j$, if  true is assigned to $x_j$, then we move a token on $s^j_0$ to $t^j_1$, otherwise, move a token on $s^j_1$ to $t^j_1$.
\item[\textbf{(M2)}] Let $c_i=(a \vee b \vee c)$. Since E3-SAT is satisfiable, at least one literal in $c_i$ is true. Let $\alpha\in \{a,b,c\}$ be one literal which is true in assignment (if there are multiple candidates, select arbitrary one). Let $L_\alpha$ be a vertex gadget corresponding to literal $\alpha$. 
By movement \textbf{\textbf{(M1)}}, if $\alpha$ is a positive literal, a token on $s^\alpha_0$ moves to $t^\alpha_1$. If $\alpha$ is a negative literal, a token on $s^\alpha_1$ move to $t^\alpha_1$. Thus, we can move a token on $v^i_0$ to $v^i_{2k}$ via $\rho_i(\alpha)$ because $\rho_i(\alpha)$ is not blocked.
\item[\textbf{(M3)}] For all vertex gadgets $L_j$, move a token that did not move in movement \textbf{(M1)} on $s^j_0$ or $s^j_1$ to $t^j_0$.
\end{description}
Note that, the total number of moves in movement \textbf{(M1)} and \textbf{(M3)} is $n$ each and the total number of moves in movement \textbf{(M2)} is $2m$.
Therefore, the total number of moves for transition from $I_s$ to $I_t$ is $2(m+n)$.
\end{proof}

Next, we focus on the only-if part. We present an auxiliary lemma.
\begin{lemma}
\label{lemma15}
Let $G$ be the graph constructed from an E3-SAT instance $\Phi$ by the reduction explained in Section~\ref{sec:gadget_construct}.
If the shortest reconfiguration sequence from $I_s$ to $I_t$ is at most $2(m+n)$ under the $k$-Jump rule, the following three statements hold  in that shortest reconfiguration sequence.
\begin{description}
    \item[\textbf{(S1)}]For any $i$ ( $0\leq i \leq m-1$ ), the token on $v^i_0$ in $I_s$ has to move exactly twice. Also, for any $j$ ( $0\leq j \leq n-1$), it is required that the tokens on the $s^j_0$ and the $s^j_1$ in $I_s$ moves exactly once.
    \item[\textbf{(S2)}]For any $i$ ( $0\leq i \leq m-1$ ), the token on $v^i_0$ in $I_s$ is placed on $v^i_{2k}$ in $I_t$. 
    \item[\textbf{(S3)}]For any $j$ ( $0\leq i \leq n-1$ ), the tokens on $s^j_0$ and $s^j_1$ in $I_s$ are placed  on $t^j_0$ or the $t^j_1$ in $I_t$. 
\end{description}
\end{lemma}
\begin{proof}
First, we show the statement \textbf{(S1)}.
For any $i$ ($0\leq i \leq m-1$), the distance between $v^i_0$ and any vertex in $I_t$ is at least $k+1$, so it is required that the token on $v^i_0$ moves at least twice.
Also, for any $j$ ($0\leq j \leq n-1$), $s^j_0$ and $s^j_1$ are not in $I_t$.
It implies that the tokens on $s^j_0$ and $s^j_1$ have to move at least once.
Since the length of the reconfiguration sequence from $I_s$ to $I_t$ is at most $2(m+n)$, the statement \textbf{(S1)} holds.

Next, we show the statement \textbf{(S2)}.
For any $0\leq i'\leq m-1$, the distance between $v^i_0$ and $v^{i'}_{2k}$ is $2k+1$ if $i \neq i'$ holds, or $2k$ otherwise.
Also, for any $0\leq j \leq m-1$, the distance from $s^j_0$ or $s^j_1$ to $v^{i}_{2k}$ is $k+1$ or $k+2$.
Therefore, from the condition of the statement \textbf{(S1)}, only a token on $v^i_0$ can move to $v^i_{2k}$.

Finally, we show the statement \textbf{(S3)}. In the reconfiguration sequence from $I_s$ to $I_t$, if a token on  $s^j_0$ (resp.~$s^j_1$) moves to the vertex that is not included in $L_j$, we call the token on $s^j_0$ (resp.~$s^j_1$) an \emph{across-gadget token}. Suppose for contradiction that there exists an across-gadget token though the reconfigure from $I_s$ to $I_t$ by at most $2(m+n)$ moves.
Let $s^{\ast}$ be the vertex in $L_j$ where the across-gadget token is placed. If there are multiple such vertices, select the vertex with the  token that moves first during reconstruction from $I_s$ to $I_t$ among the across-gadget tokens.
By the statement \textbf{(S1)}, the token on $s^{\ast}$ can only move once.
The set of vertices included in $I_t$ within distance at most $k$ from $s^{\ast}$ is $\bigcup_{0\leq j' \leq n-1}t^{j'}_0 \cup t^j_1$.
By the definition of the across-gadget token, the candidate destination for the token placed in $s^{\ast}$ is $\bigcup_{0\leq j' \leq n-1, j' \neq j}t^{j'}_0$.
Let $t^{j'}_0$ be a vertex to which the token on $s^{\ast}$ moves.
In order to move the token from $s^{\ast}$ to $t^{j'}_0$, we must move the token placed in $s^{j'}_0$ and $s^{j'}_1$ before moving the token $s^{\ast}$.
By the definition of  $s^{\ast}$, tokens on $s^{j'}_0$ and $s^{j'}_1$ are not across-gadget tokens, so they move to the vertices in $L_{j'}$.
However, they can only move to either $t^{j'}_0$ or $t^{j'}_1$, and at least one token is placed on $t^{j'}_0$.
It is a contradiction because the token placed on $s^{\ast}$ moves to $t^{j'}_0$.
\end{proof}

If no token is on $s^j_0$ or $t^j_0$, then we say that  the variable gadget $L_j$ is \emph{positively opened} for the variable $x_j$.
Otherwise, we say that  the variable gadget $L_j$ is \emph{positively closed} for the variable $x_j$.
Similarly, if there does not exists a token on $s^j_1$ and $t^j_0$, then we say that the variable gadget $L_j$ is \emph{negatively opened} for the variable $x_j$. Otherwise, we say that the variable gadget $L_j$ is \emph{negatively closed} for the variable $x_j$.
By the statement \textbf{(S3)} of lemma~\ref{lemma15}, tokens on $s^j_0$ and $s^j_1$ moves to $t^j_0$ or $t^j_1$ in one movement.
Moving a token from $s^j_0$ or $s^j_i$ to $t^j_0$ does not cause $L_j$ to become open.
Therefore, during the $2(m+n)$ token movements, $L_j$ is always either positively or negatively closed for the variable $x_j$. 

The main statement of the only-if part is the lemma below.
\begin{lemma}
\label{lemma17}
Let $G$ be the graph constructed from an E3-SAT instance $\Phi$ by the reduction explained in Section~\ref{sec:gadget_construct}.
If there exists a reconfiguration sequence from $I_s$ to $I_t$ under the $k$-Jump rule with length at most $2(m+n)$, then $\Phi$ is satisfiable.
\end{lemma}
\begin{proof}
We consider fixed reconfiguration any sequence from $I_s$ to $I_t$ with length at most $2(m+n)$.
We check if each variable gadgets $L_j$ is either positively or negatively open for the variable $x_j$ during the reconfiguration from $I_s$ to $I_t$. 
If $L_j$ is positively open for the variable $x_j$, then we assign true to $x_j$, and  if $L_j$ is negatively open for the variable $x_j$, then we assign false to $x_j$.
If $L_j$ is always both positively and negatively closed for the variable $x_j$, then we assign false to $x_j$.
We prove that this assignment satisfies $\Phi$.

Consider any clause $c_i=a\vee b \vee c$.
Let  $\rho^{-1}_{i}$ be the function such that $\rho^{-1}_{i}(k^i_0)=a$, $\rho^{-1}_{i}(k^i_1)=b$, and $\rho^{-1}_{i}(k^i_2)=c$.
If the token on $v^i_0$ reaches $v^i_{2k}$ with two movements, then it must move to $k^i_0$, $k^i_1$, or $k^i_2$.
Let $k^i_h$ be the vertex that the token passed through to go to $v^i_{2k}$. 
If $\rho^{-1}_{i}(k^i_h)$ is a positive literal, then the variable gadget corresponding to $\rho^{-1}_{i}(k_h)$ is positively open before moving a token to $k^i_h$.
Thus, the clause $c_i$ satisfies because true is assigned to the variable corresponding to $\rho^{-1}_{i}(k^i_h)$. 
Similarly, if $\rho^{-1}_{i}(k^{i}_h)$ is a negative literal, then the variable gadget corresponding to $\rho_{i}^{-1}(k^{i}_h)$ is negatively opened for the corresponding variable.
Thus, the clause $c_i$ satisfies because false is assigned to the variable corresponding to $\rho^{-1}_{i}(k^{i}_h)$.
The above argument holds for other clauses, so the E3-SAT instance $\Phi$ is satisfiable.
\end{proof}
Lemma~\ref{lemma12} is trivially deduced from Lemma~\ref{lemma14} and Lemma~\ref{lemma17}.
Finally, we prove Theorem~\ref{theorem3}.
\begin{rtheorem}{theorem3}
Let $k \geq 3$ be any integer. 
Then, there exists a graph class $\mathcal{G}_k$ such that $\mathcal{G}_k$ is a subclass of chordal graphs of diameter at most $2k + 1$ and the \ShortestIsreconf\ under the $k$-Jump rule is NP-complete for $\mathcal{G}_k$.
\end{rtheorem}
\begin{proof}
    First, we show that the \ShortestIsreconf\ under the $k$-Jump rule in the chordal graph is included in NP.
    It is known that there always exists a transformation with a polynomial number of steps under the $D(G)$-Jump rule if $G$ is a connected even-hole-free graph~\cite{KaminskiMM12}.  In Theorem~\ref{thm:SimulateTJ}, we will prove that any transformation with a polynomial number of steps under the $D(G)$-Jump rule can be converted into a transformation with a polynomial number of steps under the $3$-Jump rule. Since chordal graphs form a subclass of even-hole-free graphs, the \ShortestIsreconf\ under the $k$-Jump rule in the chordal graph is included in $NP$.
    
    Lemma~\ref{lemma10} obviously implies that $\mathcal{G}_k$ is a subclass of chordal graphs.
    In addition, Lemma~\ref{lemma12} concludes that {$k$-\ShortestIsreconf} is NP-complete. It is easy to check that the diameter of the graph for the instance of {\ShortestIsreconf} under the $k$-Jump rule obtained by the reduction from the instance of E3-SAT is $2k+1$. 
\end{proof}

\section{Concluding Remarks}
\label{sec:conclusion}
In this paper, we proposed a new reconfiguration rule for the independent set reconfiguration problem, and investigated the relationship between the value of $k$ and the computational complexity of $k$-\Isreconf.
%
We conclude this paper with some open problems related to our new rule. 
\begin{itemize}
\item The complexity of $2$-\Isreconf\ for graph classes other than split graphs: 
A major class left as an open problem is chordal graphs, as shown in Table~\ref{table:ISReconf}. 
\smallskip

\item The complexity of \textsf{\ShortestIsreconf} under the $2$-Jump rule for split graphs: Is it solvable in polynomial time?  
\smallskip

\item The approximability of \ShortestunderkIsreconf\ $(k \geq 3)$ for even-hole-free graphs: 
Using the polynomial-time algorithm by Kami\'{n}ski et al.~\cite{KaminskiMM12}, we can solve \ShortestunderkIsreconf\ for connected even-hole-free graphs $G$ when $k = D(G)$.
Does it give any non-trivial approximation factor? 

\end{itemize}

\bibliography{reference.bib}

\end{document}